\documentclass[11pt]{article}

\usepackage[T1]{fontenc}
\usepackage{lmodern}
\usepackage[margin=1.05in]{geometry}
\usepackage{amsmath,amssymb,amsthm,mathtools}
\usepackage{array,booktabs}
\usepackage{microtype}
\usepackage{needspace}
\usepackage{xcolor}
\usepackage[colorlinks=true,linkcolor=blue!55!black,
  citecolor=blue!55!black,urlcolor=blue!55!black]{hyperref}

\newtheorem{theorem}{Theorem}
\newtheorem{lemma}[theorem]{Lemma}
\newtheorem{claim}[theorem]{Claim}
\newtheorem{corollary}[theorem]{Corollary}

\newtheorem{definition}[theorem]{Definition}

\newcommand{\E}{\mathbb E}
\newcommand{\Prb}{\mathbb P}
\newcommand{\cA}{\mathcal A}
\newcommand{\F}{\mathcal F}

\newcommand{\logbinom}[2]{\log\binom{#1}{#2}}
\newcommand{\Rhat}{\widehat R}

\hypersetup{
  pdftitle={Always-Correct Succinct Dynamic Fusion Nodes Are Impossible},
  pdfauthor={Ian D'Ambrosio},
  pdfsubject={Cell-probe lower bounds for succinct dynamic dictionaries}
}

\title{\textbf{Always-Correct Succinct Dynamic Fusion Nodes Are Impossible}\\[0.35em]
\large A Cell-Probe Lower Bound in the Small-Set, Large-Universe Regime}
\author{Ian D'Ambrosio\\
\normalsize Nth Research Collective}
\date{August 19, 2026}

\begin{document}
\maketitle

\begin{abstract}
Kuszmaul, Liang, and Zhou (SODA 2026) ask whether succinct constant-time
dynamic fusion nodes exist when the number of stored keys is polylogarithmic
in the universe size. We give a negative answer for always-correct
structures. For $n^8\le U$, $\log U\ge2^{70}$, and $0\le R<n$, a dynamic
dictionary requires at least
\[
  2^{-26}\log\left(1+\frac{n}{R+1}\right)
\]
expected-amortized cell probes per operation. The model permits fixed layouts
of packed cells of at most one word each. It covers physically realized
virtual-memory/spill representations when mutable descriptors are public,
derivable, or charged as persistent state. The bound covers
deterministic structures and zero-error Las Vegas structures with fresh
per-invocation randomness.

The proof begins with the communication framework of Li, Liang, Yu, and Zhou
(FOCS 2023), whose published theorem assumes $U=n^{1+\Theta(1)}$. A direct
replay exposes a conditioning issue: after conditioning on consistency, their
inner message bound uses unconditional expected probe-set sizes. We repair
this joint by putting pointwise probe caps into the consistency event and
charging rejected true segments by Markov's inequality. The repaired argument
extends to large universes through
\[
  \gamma=\frac{\log U-\log n}{6\log U},
  \qquad \frac12\log(U/n)=3\gamma\log U.
\]
Consequently, at $U=2^w$ and $n(w)=\lceil w^c\rceil$ for any fixed real
$c>0$, no always-correct predecessor structure can use
$\lceil\log\binom Un\rceil+o(n)$ bits and support worst-case $O(1)$
operations. Constant expected-amortized time requires $\Omega(n)$
redundant bits.
\end{abstract}

\section{Introduction}

A dynamic fusion node stores a small ordered set and supports predecessor and
updates in constant time. Such nodes are basic components of fast integer data
structures. Their known implementations, however, use substantially more
than the information-theoretic minimum number of bits. Kuszmaul, Liang, and
Zhou~\cite[Section~7]{KLZ26} isolate the resulting question:

\begin{quote}\itshape
Whether it is possible to build succinct $O(1)$-time dynamic fusion nodes
remains open.
\end{quote}

Li, Liang, Yu, and Zhou (LLYZ)~\cite{LLYZ23} proved a sharp dynamic-dictionary
tradeoff at polynomial universes. If $U=n^{1+\Theta(1)}$ and the redundancy
is $R<n$, then some operation requires
$\Omega(\log(n/R))$ expected-amortized probes. The stated universe hypothesis
excludes the fusion-node regime $n=\operatorname{polylog}U$. Subsequent work
continued to treat the large-universe case as open~\cite{KLZ26,D26}.

Our main result removes that barrier, after repairing a conditioning joint in
the inherited communication proof.

\paragraph{Contributions.}
First, we prove an explicit cell-probe lower bound throughout the
small-set, large-universe regime $n^8\le U$. This gives a negative answer to
the dynamic-fusion-node question for exact deterministic and Las Vegas
structures with sublinear redundancy. Second, we isolate and repair the
conditioning defect in the inherited inner communication game, then show that
the large-universe extension is controlled by the single parameter
$\gamma$ in~\eqref{eq:gamma}. Third, we formalize the full argument in
Lean~4~\cite{Lean4}: the packed-memory model, hard distribution, communication
protocols, separator, forest accounting, deterministic and randomized bounds,
strict predecessor reduction, and redundancy corollaries are all checked by
the kernel.

\begin{theorem}[Main theorem]\label{thm:main}
Let $n,U$ be natural numbers with $n^8\le U$ and $\log U\ge2^{70}$. Every
always-correct deterministic or Las Vegas dynamic dictionary with capacity
$n$, word size $w=\lceil\log U\rceil$, and redundancy $0\le R<n$ has
expected-amortized probe complexity at least
\[
  2^{-26}\log\left(1+\frac{n}{R+1}\right).
\]
The assumptions $0\le R<n$ imply $n\ge1$. For a Las Vegas dictionary, fix an
input-independent probability law $\mu$ on complete fresh-random schedules and
draw the hard instance independently: $(\mathcal I,\omega)\sim
\mathcal D_1\times\mu$. The normalization is by the exact $4n$ primitive
operations in the distribution.
\end{theorem}

\begin{theorem}[Fusion-node consequence]\label{thm:fusion}
Let $U=2^w$ and $n(w)=\lceil w^c\rceil$ for any fixed real $c>0$. For
positive integer $c$, one may instead take $n(w)=w^c$ exactly. An
always-correct physical word-RAM structure supporting insert, delete, and
predecessor on at most $n(w)$ keys,
using
\[
  \left\lceil\log\binom Un\right\rceil+o(n)
\]
persistent mutable bits, requires $\omega(1)$ expected-amortized time. In
particular, it cannot have worst-case $O(1)$ time. The canonical strict
predecessor reduction queries $x+1$ in $[U]$ and applies the main theorem on
$[U-1]$. If redundancy is measured against ${U\choose n}$, the reduced
membership redundancy is at most $R+1$ when $n\le U/2$, so the explicit
strict predecessor denominator is $R+2$. This one-bit transfer does not change
the $\omega(1)$ conclusion. All persistent physical contents, including nested
adapter metadata, count in the displayed space bound. If the final physical
memory is a variable prefix, Lemma~\ref{lem:physicalization} adds only
$O(\log S)=o(n)$ bits for its length. The fixed-layout denominator then becomes
$R+O(\log S)$, which leaves the conclusion unchanged.
\end{theorem}

\begin{corollary}[Linear redundancy]\label{cor:linear}
In the regime of Theorem~\ref{thm:main}, expected-amortized cost at most $T$
implies, with $C_0:=2^{26}$,
\[
  R+1\ge n\cdot2^{-C_0T}.
\]
In particular, constant expected-amortized time requires $R=\Omega(n)$.
\end{corollary}

\begin{corollary}[Sublinear-redundancy time scale]\label{cor:scale}
At $U=2^w$ and $n=w^c$ for a fixed integer $c\ge1$, redundancy
$R\le n\cdot2^{-a\sqrt{\log w}}$ for fixed $a>0$ forces
expected-amortized cost at least
$(a/2^{27})\sqrt{\log w}$ for all sufficiently large $w$.
\end{corollary}

Domingues~\cite[Theorem~3]{D26} gives a rank index at the redundancy and time
scale in Corollary~\ref{cor:scale}. His reduction to a full fully indexable
dictionary incurs an additional $O(n\log w)$ bits~\cite[Lemma~2]{D26}. Thus
our lower bound matches the index's time coordinate, but it does not assert
that a full dictionary at the matching point already exists.

Two August 2026 results address adjacent compression models.
Blelloch et al.~\cite{BHKLZ26a} give constant-time dynamic entropy-encoded
arrays, while Blelloch et al.~\cite{BHKLZ26b} determine tradeoffs relative to
gap entropy for difference-encoded membership dictionaries. Neither supplies a
full dynamic predecessor structure at $\log\binom Un+o(n)$ bits in the
fusion-node regime.

\paragraph{The proof repair.}
The LLYZ inner game conditions on an event $W$ saying that two executions are
consistent. It then bounds variable messages using the unconditional
hypotheses $\E|S_A|,\E|S_B|\le mt$. Conditioning need not preserve these
bounds. We instead define $W$ to include the pointwise caps
$|S_A|,|S_B|\le mt$. The outer game accepts an alternative execution only if
it obeys the same cap. Markov's inequality shows that few true segments are
lost, while the inner message lengths become valid pointwise after
conditioning. This is the only structural change to the communication game.
Our statement concerns a gap in the written conditional message-length
estimate, not a counterexample to the published polynomial-universe theorem.

\paragraph{The universe extension.}
Set
\begin{equation}\label{eq:gamma}
  \gamma:=\frac{\log U-\log n}{6\log U}.
\end{equation}
For $n^8\le U$, $\gamma\ge7/48$. The four binomial estimates in the
inner game depend on the universe ratio through
\begin{equation}\label{eq:identity}
  \frac12\log(U/n)=3\gamma\log U.
\end{equation}
No step requires the ratio $\log U/\log n$ to be bounded above.

\paragraph{Scope.}
The result applies to exact deterministic and Las Vegas dictionaries. It does
not cover Monte Carlo dictionaries with erroneous answers. Static lookup
tables independent of the maintained set and update history may be shared for
free. Las Vegas randomness may be freshly sampled on every invocation; no
random-tape cursor or operation counter is charged or stored. Every persistent
mutable bit counts toward the redundancy.

\paragraph{Machine check.}
The complete fixed-layout model, fresh-random indexed execution, hard
distribution, communication bounds, nested-forest accounting, deterministic
and Las Vegas lower bounds, canonical strict predecessor reduction, and stated
corollaries have been checked by the Lean~4 kernel. The promoted declarations
use no custom axioms. Section~\ref{sec:formal} gives the declaration map and
replay command.

\paragraph{Proof architecture.}
Distribution~1 supplies a symmetric sequence of delete-insert windows. The
capped inner game shows that two sufficiently different insertion sets are
exponentially unlikely to produce compatible executions while both obeying a
pointwise probe budget. The outer game turns that statement into a dichotomy:
each window either spends many probes or contributes many cross-segment cell
reuses. A single padded hierarchy assigns every selected reuse to the unique
lowest common ancestor of consecutive occurrences of the same cell. Summing
over the hierarchy converts the local dichotomy into the global logarithmic
lower bound. The rest of the paper proves these four steps in that order.

\section{Model}\label{sec:model}

All logarithms are base two and $w=\lceil\log U\rceil$.

\begin{definition}[Fixed-layout packed cell probes]\label{def:model}
At capacity $n$, persistent memory has a public fixed finite list of cells.
Each cell has positive width at most $w$ bits, and the total width of the
layout is an integer $B$. A probe reads or writes one cell. Computation is free.

Every memory string has one position-independent represented set
$\operatorname{Rep}(C)\subseteq[U]$. Invocation position may select a fixed
random stream after conditioning, but it never changes the interpretation of
an unchanged memory string.

All persistent mutable bits, including seeds, pointers, and allocator state,
count toward $B$. Transient scratch can be free only if it begins each
operation in a public state and carries no information to the next operation;
its accesses still count as probes. A static table independent of the set and
operation history can be shared for free.
\end{definition}

A randomized operation may draw fresh private bits adaptively. After fixing
one local random stream for each global invocation position, its persistent
behavior is deterministic but need not be stationary: the same operation from
the same memory may have different traces at two positions. The invocation
position selects the fixed stream only in the analysis. It is not stored in
memory and is not available as algorithmic advice.

The model contains ordinary $w$-bit word memory and final physical
virtual-memory/spill representations~\cite[Sections~5.7 and~6]{KLZ26}. KLZ's
internal VMs and adapters are logical objects; their word accesses are
translated until the complete persistent state is one physical root prefix.
The lower bound simulates those final physical accesses, not the internal
translation steps. Therefore every nested length, spill, chunk map, and adapter
descriptor that survives between operations is already part of the claimed
physical space. An abstract VM interface that supplies content-dependent
metadata externally for free is not a physical word-RAM representation and is
not covered.

\begin{lemma}[Physical-prefix compilation]\label{lem:physicalization}
Suppose the complete persistent state of a word-RAM dictionary is one physical
prefix of at most $M_{\max}$ $w$-bit words plus $O(1)$ root spills, and its
total stored contents---including every nested-adapter length, spill descriptor,
chunk map, and address-translation parameter---use at most $S$ bits. Then the
dictionary embeds in a fixed packed layout of
\[
  S+\left\lceil\log(M_{\max}+1)\right\rceil+O(1)
\]
bits. Every final physical word/spill probe and every prefix resize is simulated
with $O(1)$ packed-cell probes, provided
$\lceil\log(M_{\max}+1)\rceil=O(w)$.
\end{lemma}
\begin{proof}
The state map writes $M$ in a fixed
$\lceil\log(M_{\max}+1)\rceil$-bit field, followed by the $M$ physical words,
the $O(1)$ root spills, and zero padding to the public maximum. This map is
injective because $M$ and the complete physical contents are recovered
exactly. All nested metadata is already among those contents; it is not an
extra input to the simulator.

A final physical word access at address $i\le M$ reads or writes the
corresponding fixed word cell. Root spills use their fixed public spill slots
after the maximum word region. Before allocation or release, the charged
program reads the length cell and its bounded root-spill migration trace.
The resize contract records whether the length is stable, increases by one, or
decreases by one. It also requires every migration address to be a root-spill
cell and preserves its stored value across the transition. On allocation, the
physical transition writes the new word before the word carries information.
On release, the charged body writes any moved spill and the resize contract
requires zero padding in every vacated word slot. These writes are part of the
charged body, not free allocator state. Updating $M$ costs one additional cell probe. The
length field uses $O(1)$ packed cells under the displayed premise. Thus every
final physical operation has constant-factor probe overhead and the displayed
space bound. KLZ's Lemma~6.1 is an explicit example that stores all logical VM
lengths and chunk mappings inside such a final physical prefix.
\end{proof}

If the physical representation uses
$\log\binom Un+\rho$ bits, then $S=\log\binom Un+\rho+O(1)$ and
$\log M_{\max}=O(\log S)=O(w)$. At $U=2^w$ and $n=w^c$,
$S=\Theta(nw)$ and $O(\log S)=o(n)$ for every fixed $c>0$. Hence
$\rho=o(n)$ remains $o(n)$ after compilation.

A dictionary maintains $S\subseteq[U]$, $|S|\le n$, under exact
$\operatorname{Insert}$, $\operatorname{Delete}$, and membership
$\operatorname{Query}$. Its redundancy is
\begin{equation}\label{eq:redundancy}
  R:=B-\log\binom Un\ge0.
\end{equation}
An arbitrary $B$-bit string in the layout need not be reachable. We use one
such mixed string in the inner game, but never execute the original data
structure from it.

\section{Hard distribution}\label{sec:distribution}

We use Distribution~1 of LLYZ~\cite{LLYZ23}. Choose a uniform $n$-subset
$K_0\subseteq[U]$ and insert its keys in increasing order. Then, for
$i=1,\ldots,n$:
\begin{enumerate}
  \item choose $d_i$ uniformly from the keys of $K_0$ not previously removed;
  \item execute $\operatorname{Query}(d_i)$ and
        $\operatorname{Delete}(d_i)$;
  \item choose $a_i$ uniformly outside both $K_0$ and the current set, and
        execute $\operatorname{Insert}(a_i)$.
\end{enumerate}
The query-delete-insert triple is one \emph{meta-operation}. The sequence has
$4n$ primitive operations: $n$ initializing inserts and $3n$ operations in
the meta-phase.

At the start of any window of $m$ meta-operations, write $K$ for the current
$n$-set, $D$ for the keys deleted in the window, and $A$ for the inserted
keys. The symmetry facts used below are:
\begin{enumerate}
  \item $K$ is a uniform $n$-subset of $[U]$;
  \item conditioned on $K$, $D$ is a uniform $m$-subset of $K$, with a
        uniform order;
  \item conditioned on $K$, $A$ is a uniform $m$-subset of $[U]\setminus K$,
        with a uniform order; and
  \item $(K\setminus D)\cup A$ is a uniform $n$-subset.
\end{enumerate}
For the first two claims, hide the age labels of the current keys. The
surviving original keys form a uniform subset of the current set, and the
marginal deletion set is uniform. For insertions, the removed original keys
form a uniform hidden forbidden subset outside the current set. Every
candidate $A$ has the same number of disjoint hidden forbidden supersets.
Deleting uniformly from all currently present keys would be a different
distribution and would not imply these facts.

The complete joint law first samples $K$, then the uniform ordered deletion
set, and then the ordered insertion sets with the stated conditional laws.
The public relative permutations are independent of these semantic sets. No
independence is assumed for the concrete starting memory
$C_{\mathrm{bef}}$: it may be any reachable $B$-bit memory representing $K$,
and any information it carries is charged when the mixed $B$-bit state is sent.

\section{Tree accounting and the outer lemma}\label{sec:outer}

Partition an interval of $M$ consecutive meta-operations into $\lambda$
segments of $m$ operations, where $M=\lambda m$. Assign a later probe of a
cell to the lowest tree node whose distinct child segments contain that probe
and the preceding probe of the same cell. Let $\operatorname{cost}_u$ count
probes assigned to node $u$, and let $\operatorname{probe}_u$ count every
probe during its interval. A probe is assigned at most once; intervals at one
level are disjoint.

We use a public-coin Bloomier-style separator inspired by Chazelle, Kilian,
Rubinfeld, and Talwar~\cite{CKRT04}. An independent public seed $Z$ is a uniform
ordering of all Boolean predicates on the finite address universe. For disjoint
sets $P,N$, Alice sends the self-delimiting index of the first predicate that
is true on $P$ and false on $N$; Bob recovers the predicate from $(Z,\text{index})$.
Every seed contains a separator, and averaging over $Z$ gives expected length
at most $2(|P|+|N|)+4$ bits, independent of the address-universe size. The seed
is independent of the hard instance and shared for free. In the regimes below,
$|P|,|N|\le mt$ and $mt\ge2$, so this cost is at most $6mt$.

\begin{lemma}[Capped outer lemma]\label{lem:outer}
Let $u$ span $M=\lambda m$ meta-operations and let $t>0$, with $m\ge4$ and
$mt\ge2$. Suppose
\begin{align}
  2^{64}&\le\lambda\le\gamma\log U-5, \label{outer:a}\\
  M\log\lambda&\ge100(R+1), \label{outer:b}\\
  R+7m+8mt&\le\tfrac12\gamma m\log U, \label{outer:c}\\
  n^{1-\gamma/2}&\le m<n/2. \label{outer:d}
\end{align}
Then
\[
  \E[\operatorname{cost}_u]\ge\frac{\gamma}{100}M
  \quad\text{or}\quad
  \E[\operatorname{probe}_u]\ge\frac{Mt}{16}.
\]
\end{lemma}

\subsection{Outer communication game}

Assume that both conclusions of Lemma~\ref{lem:outer} fail. Call a cell
\emph{bad} if it is probed in at least two segments of $u$, and write
$S_{\mathrm{bad}}$ for all such cells. Let $S_A^{[i]}$ be the cells probed
by the true execution of segment $i$. A cell appearing in $s\ge2$ segment
probe sets contributes at least $s-1\ge s/2$ assigned probes. Hence
\begin{equation}\label{eq:bad}
  \frac12\E\left[\sum_i
  |S_A^{[i]}\cap S_{\mathrm{bad}}|\right]
  \le\E[\operatorname{cost}_u]<\frac{\gamma M}{100}.
\end{equation}

Both players know the starting state $C_{\mathrm{st}}$, the unordered set
$\cA$ of all $M$ insertions in $u$, the deletion sequence, and one relative
permutation $\pi_i\in S_m$ for each segment. The permutation specifies the
order after that segment's unknown $m$-set is decoded; it reveals no key
identity or segment membership. Alice additionally knows the partition
$(A^{[1]},\ldots,A^{[\lambda]})$ of $\cA$ and must send it to Bob.

Alice first encodes the final state $C_{\mathrm{end}}$ conditioned on the
public input. The final key set $T$ is separately determined by
$C_{\mathrm{end}}$ and by the public input, and it is uniform among
$n$-subsets. Since
$H(C_{\mathrm{end}})\le B=\logbinom Un+R$,
if $P$ denotes the public input and $T$ the final key set, then
\[
  I(C_{\mathrm{end}};P)\ge I(T;P)=H(T)=\logbinom Un.
\]
Consequently,
\begin{equation}\label{eq:endstate}
  H(C_{\mathrm{end}}\mid P)
  =H(C_{\mathrm{end}})-I(C_{\mathrm{end}};P)\le R.
\end{equation}
A prefix code has expected length at most $R+1$.

Bob processes segments from left to right. Before segment $i$, previously
decoded parts determine its start state $C_{\mathrm{bef}}$. Let
\[
  \F_{\mathrm{all}}^{[i]}
  :=\{Q\subseteq\cA_{\mathrm{rest}}:|Q|=m\}.
\]
Alice defines
\begin{equation}\label{eq:xdef}
  X_i:=\mathbf1\left[
    |S_A^{[i]}\cap S_{\mathrm{bad}}|\le\gamma m/2
    \ \text{and}\ |S_A^{[i]}|\le mt
  \right].
\end{equation}
If $X_i=0$, she sends the rank of $A^{[i]}$ in
$\F_{\mathrm{all}}^{[i]}$. If $X_i=1$, Bob tests every
$Q\in\F_{\mathrm{all}}^{[i]}$: execute the segment from
$C_{\mathrm{bef}}$, obtaining probe set $S_Q$ and ending string $C_Q$, and
accept exactly when
\[
  |S_Q|\le mt,
  \qquad
  |\{z\in S_Q:C_Q(z)\ne C_{\mathrm{end}}(z)\}|\le\gamma m/2.
\]
The true set passes. A true-run cell altered after segment $i$ must be probed
in another segment and is therefore bad.

Let $\F_{\mathrm{qual}}^{[i]}$ be the accepted family and
\[
  \F_{\mathrm{nbr}}^{[i]}
  :=\{Q:|A^{[i]}\cap Q|\ge m/2\},
  \qquad
  Y_i:=\mathbf1[\F_{\mathrm{qual}}^{[i]}
                  \subseteq\F_{\mathrm{nbr}}^{[i]}].
\]
The proof charges an ideal adaptive rank cost, not a separately rounded code
at every stage. Conditional on the cap bit, stage $i$ contributes
$\log|\F_{\mathrm{all}}^{[i]}|$ or
$\log|\F_{\mathrm{qual}}^{[i]}|$. The product of the corresponding inverse
family sizes satisfies Kraft's inequality. We use its entropy converse
directly, without constructing or rounding a literal rank transcript. The cap
vector itself costs exactly $\lambda$ bits. Writing
$c_i$ for the ideal rank cost and charging its cap bit gives
\begin{equation}\label{eq:segment-message}
 c_i+1\le\log|\F_{\mathrm{all}}^{[i]}|
 -X_iY_i\bigl(\log|\F_{\mathrm{all}}^{[i]}|
                -\log|\F_{\mathrm{nbr}}^{[i]}|\bigr)+1.
\end{equation}

\subsection{Three estimates}

\begin{claim}[Many capped true segments]\label{cl:many-x}
\[
  \E\left[\sum_{i\le\lambda/2}X_i\right]
  \ge\frac{159}{400}\lambda>\frac\lambda4.
\]
\end{claim}
\begin{proof}
Equation~\eqref{eq:bad} bounds the expected number of segments failing the
bad-cell cap by $\lambda/25$. Also,
\[
  \sum_i\Prb[|S_A^{[i]}|>mt]
  \le\frac{\E[\operatorname{probe}_u]}{mt}<\frac\lambda{16}.
\]
Subtracting both failure counts from the first $\lambda/2$ segments gives
$\lambda/2-\lambda/25-\lambda/16=159\lambda/400$.
\end{proof}

\begin{claim}[Far qualified sets are rare]\label{cl:far}
For every $i\le\lambda/2$,
\[
  \Prb[Y_i=0\ \text{and}\ X_i=1]\le\frac18.
\]
\end{claim}
\begin{proof}
If $X_i=1$ and a far $Q$ qualifies, then the true and alternative executions
are consistent on $S_A\cap S_Q$ and both probe at most $mt$ cells. The capped
inner lemma, Lemma~\ref{lem:inner}, bounds this event for every fixed
intersection size $g<m/2$ by $U^{-\gamma m}$.

The distribution match is exact. Conditional on $g$, every pair of $m$-sets
$A,Q\subseteq[U]\setminus K$ with $|A\cap Q|=g$ is contained in the same
number
\[
  \binom{|[U]\setminus K|-(2m-g)}
        {|\cA_{\mathrm{rest}}|-(2m-g)}
\]
of possible residual supersets. Thus sampling the residual set first and then
$Q$ produces the pair law in Lemma~\ref{lem:inner}. A union bound gives
\[
  |\F_{\mathrm{all}}^{[i]}|U^{-\gamma m}
  \le2^M U^{-\gamma m}
  =2^{m(\lambda-\gamma\log U)}
  \le2^{-5m}\le\frac18
\]
by~\eqref{outer:a}.
\end{proof}

\begin{claim}[Counting gap]\label{cl:gap}
For $i\le\lambda/2$,
\[
  \log|\F_{\mathrm{all}}^{[i]}|
  -\log|\F_{\mathrm{nbr}}^{[i]}|
  \ge\frac m4\log\lambda.
\]
\end{claim}
\begin{proof}
At least $M/2$ insertions remain, so
\[
  \log|\F_{\mathrm{all}}^{[i]}|
  \ge\log\binom{M/2}{m}\ge m(\log\lambda-1).
\]
Counting neighbors by intersection size gives
\[
  \log|\F_{\mathrm{nbr}}^{[i]}|
  \le4m+\frac m2\log\lambda.
\]
The difference is at least $(m/4)\log\lambda$ because
$\log\lambda\ge64$.
\end{proof}

Claims~\ref{cl:many-x} and~\ref{cl:far} imply
\begin{equation}\label{eq:xy}
 \sum_{i\le\lambda/2}\E[X_iY_i]
 \ge\frac{159}{400}\lambda-\frac\lambda2\cdot\frac18
 =\frac{134}{400}\lambda>\frac\lambda4.
\end{equation}
Summing~\eqref{eq:segment-message}, applying Claim~\ref{cl:gap}, and
telescoping the partition counts yields the total ideal rank cost plus the
cap-vector bits:
\begin{equation}\label{eq:outer-upper}
 \lambda+\sum_i\E c_i
 \le\lambda+\log M!-\lambda\log m!-\frac{M}{16}\log\lambda.
\end{equation}
The partition has entropy $\log M!-\lambda\log m!$. Accounting for the
expected $R+1$ bits used for $C_{\mathrm{end}}$, the Kraft inequality gives
the matching lower bound directly, with no per-stage ceiling:
\begin{equation}\label{eq:outer-lower}
 \lambda+\sum_i\E c_i\ge\log M!-\lambda\log m!-(R+1).
\end{equation}
Therefore
\[
 R+1\ge\frac{M}{16}\log\lambda-\lambda
 \ge\frac{3M}{64}\log\lambda
 >\frac{M}{100}\log\lambda,
\]
using $\lambda\le M$ and $\log\lambda\ge64$. This contradicts
\eqref{outer:b} and proves Lemma~\ref{lem:outer}, once the inner lemma is
established.

\section{Capped inner game}\label{sec:inner}

Fix one segment. Let $m<n/2$ and $g<m/2$. Sample a uniform $n$-set
$K\subseteq[U]$, a uniform $m$-subset $D\subseteq K$ with order $\sigma$,
and uniform $m$-subsets $A,Q\subseteq[U]\setminus K$ conditioned on
$|A\cap Q|=g$. A common permutation $\pi$ determines both insertion orders.
Starting from the random state $C_{\mathrm{bef}}$ representing $K$, let
$C_A,C_Q$ be the two ending strings and $S_A,S_Q$ their probe sets.

\begin{definition}[Capped consistency]\label{def:consistency}
The two executions are capped-consistent when
\begin{enumerate}
  \item $C_A(S_A\cap S_Q)$ and $C_Q(S_A\cap S_Q)$ differ on at most
        $\gamma m$ cells; and
  \item $|S_A|\le mt$ and $|S_Q|\le mt$.
\end{enumerate}
\end{definition}

\begin{lemma}[Capped inner lemma]\label{lem:inner}
Suppose
\begin{align*}
  4&\le m,\\
  m&\le n\le U^{1/2},\\
  m&\ge n^{1-\gamma/2},\\
  mt&\ge2,\\
  \gamma\log U&\ge25,\\
  R+7m+8mt&\le\tfrac12\gamma m\log U.
\end{align*}
Then, for every $g<m/2$,
\[
 \Prb[A,Q\text{ are capped-consistent}\mid|A\cap Q|=g]
 \le U^{-\gamma m}.
\]
\end{lemma}

\subsection{Protocol}

Fix arbitrary values of the public permutations $(\pi,\sigma)$, and let $Z$
be the independent uniform public separator seed described above. Before
conditioning on success, $(K,D,A,Q)$ retains the uniform product law below,
independently of $Z$. We prove the probability bound for every fixed
permutation pair and then average.

Let $W$ indicate capped consistency. Alice knows $C_{\mathrm{bef}}$ and the
ordered sequences; Bob knows the fixed $\pi$, $\sigma$, and public seed $Z$.
Conditional on $W=1$, Alice sends enough information for Bob to recover
$(A,Q,D,K)$.

\begin{enumerate}
  \item Send $Q$ in $\logbinom Um$ bits.
  \item Send the mixed $B$-bit string $C_{\mathrm{mix}}$: on $S_Q$ it equals
        $C_{\mathrm{bef}}$, and outside $S_Q$ it equals $C_A$.
  \item Send the public-seed separator index describing a superset $S_Q^*$ of
        $S_Q$ whose complement contains $S_A\setminus S_Q$. Conditional on the
        probe sets, its expectation over $Z$ is at most
        $4mt+4\le6mt$ bits.
  \item Recover $D$. Bob enumerates candidate $m$-sets, simulates the
        $Q$-execution using the decoded values on $S_Q^*$, and rejects a
        candidate if its next probe leaves $S_Q^*$ or a query misses. The true
        $D$ passes. Every survivor lies in $K\cup Q$, so Alice's index costs at
        most $\logbinom{n+m}{m}$ bits.
  \item Send $S_A\cap S_Q$ as a subset of the now-known $S_Q$, using at most
        $mt$ bits.
  \item Send the cells where $C_A$ and $C_Q$ differ on the intersection: an
        address mask of at most $mt$ bits and at most $\gamma mw$ content
        bits.
  \item Bob reconstructs $C_A$: use the corrected $C_Q$ values on
        $S_A\cap S_Q$ and $C_{\mathrm{mix}}$ everywhere else. The exact
        dictionary state determines $(K\setminus D)\cup A$. Alice identifies
        $A$ inside this $n$-set using $\logbinom nm$ bits, after which Bob
        recovers $K$.
\end{enumerate}

The separator is sufficient for the partial simulation. It includes $S_Q$.
For an extra cell in $S_Q^*\setminus S_Q$, the separator excludes
$S_A\setminus S_Q$; hence either the cell is outside both probe sets, where
$C_{\mathrm{bef}}=C_A=C_{\mathrm{mix}}$, or it lies in the intersection
already included in $S_Q$. Thus Bob knows the required starting content on
all of $S_Q^*$. The true $D$ never probes outside $S_Q\subseteq S_Q^*$.

The three finite-set ranks are encoded jointly and introduce fewer than three
rounding bits in total; all masks have their displayed fixed lengths. Since
$m\ge4$, the reserve of $m$ bits covers every ceiling and delimiter. Averaging
the separator over the independent $Z$, the pointwise caps give
\begin{align}
\E_{z\sim Z}\!\left[
  H(M_{\mathrm{send}}\mid W=1,\pi,\sigma,Z=z)\right]\le{}&
 \logbinom Um+\bigl(\logbinom Un+R\bigr)+6mt
 +\logbinom{n+m}{m}\notag\\
&+mt+mt+\gamma mw+\logbinom nm+m.
\label{eq:send}
\end{align}
This remains valid after conditioning on $W=1$, which is the purpose of the
capped event.

\subsection{Entropy comparison}

For each fixed permutation pair, the learned tuple is a deterministic function
of the transcript and Bob's fixed side information. Averaging over $Z$ gives
\[
 \E_{z\sim Z}\!\left[
   H(M_{\mathrm{learn}}\mid W=1,\pi,\sigma,Z=z)\right]
 \le
 \E_{z\sim Z}\!\left[
   H(M_{\mathrm{send}}\mid W=1,\pi,\sigma,Z=z)\right]
\]
whenever the conditioning event has positive probability; if its probability
is zero, the desired bound is immediate.

For the fixed $(\pi,\sigma)$, before conditioning on $W$, each tuple
$(K,D,A,Q)$ has probability
\[
 \left[
  \binom Un\binom nm\binom{U-n}{m}
  \binom mg\binom{U-n-m}{m-g}
 \right]^{-1}.
\]
Since $Z$ is independent of both the tuple and $W$, conditioning on it changes
neither this law nor the learned-tuple entropy. Write
$p_{\pi,\sigma}:=\Prb[W=1\mid\pi,\sigma]$. Conditioning can increase
a tuple's probability by at most $1/p_{\pi,\sigma}$.
Therefore, after averaging over $Z$,
\begin{align}
\E_{z\sim Z}\!\left[
  H(M_{\mathrm{learn}}\mid W=1,\pi,\sigma,Z=z)\right]\ge{}&
 \logbinom Un+\logbinom nm+\logbinom{U-n}{m}
 +\logbinom mg\notag\\
&+\logbinom{U-n-m}{m-g}+\log p_{\pi,\sigma}.
\label{eq:learn}
\end{align}
Combining~\eqref{eq:send} and~\eqref{eq:learn} gives
\begin{align}
\log\frac1{p_{\pi,\sigma}}\ge{}&
 \logbinom mg+\logbinom{U-n-m}{m-g}
 -\left(\logbinom Um-\logbinom{U-n}{m}\right)\notag\\
&-\logbinom{n+m}{m}-R-8mt
 -\gamma mw-m.
\label{eq:entropy-core}
\end{align}

We need four estimates.
\begin{enumerate}
  \item Since $m,n\le U^{1/2}$ and $\log U\ge2^{70}$,
  \begin{equation}\label{eq:bin-a}
    \logbinom Um-\logbinom{U-n}{m}\le m
  \end{equation}
  Expand the ratio as a product. Since $U-n-m\ge U/2$,
  \[
    \sum_{j<m}\log\left(1+\frac{n}{U-n-j}\right)
    \le\frac{4mn}{U}\le m.
  \]

  \item Since $U-n-m\ge U/2$, $m-g>m/2$, and $m\le n$,
  \begin{align}
    \logbinom{U-n-m}{m-g}
    &\ge\frac m2\log\frac Un-m\notag\\
    &=3\gamma m\log U-m,
    \label{eq:bin-b}
  \end{align}
  where the equality is~\eqref{eq:identity}.

  \item From $m\ge n^{1-\gamma/2}$,
  \begin{equation}\label{eq:bin-c}
    \logbinom{n+m}{m}
    \le m\left(3+\log\frac nm\right)
    \le\frac\gamma2m\log U+3m.
  \end{equation}

  \item $\logbinom mg\ge0$.
\end{enumerate}
Since $w=\lceil\log U\rceil\le\log U+1$ and $\gamma\le1/6$,
substitution into~\eqref{eq:entropy-core} yields
\begin{align*}
 \log\frac1{p_{\pi,\sigma}}
 &\ge\frac32\gamma m\log U
       -\bigl(7m+R+8mt\bigr)\\
 &\ge\gamma m\log U.
\end{align*}
Thus $p_{\pi,\sigma}\le U^{-\gamma m}$ for every fixed public permutation
pair. Averaging over $(\pi,\sigma)$ gives $\Prb[W=1]\le U^{-\gamma m}$,
proving Lemma~\ref{lem:inner} and Lemma~\ref{lem:outer}.

\section{Proof of the main theorem}\label{sec:main-proof}

Fix $\lambda=2^{64}$ and assume the explicit threshold $\log U\ge2^{70}$.
It implies
\[
 \lambda\le\frac7{48}\log U-5,
 \qquad \gamma\log U\ge\frac7{48}2^{70}>2^{64}+5.
\]
Set
\begin{equation}\label{eq:t-m0}
 t:=\frac{\gamma\log U}{32},
 \qquad \Rhat:=R+1,
 \qquad m_0:=\max\left\{\lceil\Rhat\rceil,
       \left\lceil n^{1-\gamma/2}\right\rceil\right\}.
\end{equation}
If $\Rhat\ge n/2$, the logarithm in Theorem~\ref{thm:main} is $O(1)$ and
the set-changing probe floor below proves the result. Otherwise continue with
the forest construction.

For
\[
 L:=\log_\lambda\left(\left\lfloor n/m_0\right\rfloor\right),
\]
where the right side is the integer logarithm, and every $j=1,\ldots,L$,
partition the meta-operation sequence into full
nodes of length
\[
 M_j:=m_0\lambda^j,
\]
each with $\lambda$ children of length $m_0\lambda^{j-1}$. The accounting
object is one padded complete hierarchy of depth $L+1$. Out-of-range base
leaves are empty; partial nodes are traversed but not charged, preserving
smaller full nodes inside larger-level leftovers. Only the first $L$ levels
enter the outer-lemma dichotomy. Each selected noninitial probe occurrence is
assigned to the unique lowest common ancestor of it and the immediately
preceding occurrence of the same cell.

Write $x:=\log(n/\Rhat)$. The depth estimate below gives
$L\ge\gamma x/256$ whenever $x\ge8192$. Thus $L=0$ can occur only in the
small-logarithm case handled by the set-changing probe floor below. For the
forest argument assume $L\ge1$; then $m_0\le n/\lambda<n/2$.

Each full node satisfies Lemma~\ref{lem:outer}. Conditions
\eqref{outer:a}, \eqref{outer:b}, and \eqref{outer:d} follow from the
definitions. The growth condition, $\gamma\le1/6$, and $m<n/2$ imply
$m\ge4$ whenever such a node exists. Moreover
$t=\gamma\log U/32>2^{59}$, so $mt\ge2$. For a child of length $m\ge m_0$,
$R\le m$ and
\[
 8mt=\frac\gamma4m\log U.
\]
Since $R+7m\le8m\le(\gamma/4)m\log U$, condition~\eqref{outer:c} also
holds.

At any fixed level, at least half the full nodes satisfy the probe conclusion
or at least half satisfy the cost conclusion. If some level is probe-heavy,
$\lfloor n/M_j\rfloor\ge n/(2M_j)$, so the meta-phase uses at least
\[
 \frac12\frac{n}{2M_j}\frac{M_jt}{16}=\frac{nt}{64}
\]
probes. Dividing by all $4n$ primitive operations gives $t/256$. Since
$x\le\log n$, $\log U\ge8\log n$, and $\gamma\ge7/48$,
\begin{equation}\label{eq:probe-constant}
 \frac{t}{256}
 =\frac{\gamma\log U}{32\cdot256}
 \ge\frac{7x}{192\cdot256}
 \ge 2^{-25}x.
\end{equation}

Otherwise every level is cost-heavy. Each level contributes at least
\[
 \frac12\frac{n}{2M_j}\frac{\gamma M_j}{100}
 =\frac{\gamma n}{400}
\]
assigned probes. Assignments across levels are disjoint, so the average over
the full sequence is at least $\gamma L/1600$.

Suppose $x\ge8192$. Rounding changes $m_0$ by at most a factor two. If
$\Rhat\ge n^{1-\gamma/2}$, then
\[
 L\ge x/64-2\ge x/128.
\]
If $\Rhat<n^{1-\gamma/2}$, then $x\le\log n$ and
\[
 L\ge\gamma\log n/128-2\ge\gamma x/256.
\]
The weaker common estimate $L\ge\gamma x/256$ therefore holds in both
cases. In the cost-heavy branch,
\begin{equation}\label{eq:cost-constant}
 \frac{\gamma L}{1600}
 \ge\frac{\gamma^2x}{256\cdot1600}
 \ge\frac{(7/48)^2x}{256\cdot1600}
 \ge2^{-25}x.
\end{equation}
Moreover, $x\ge8192$ implies
$\log(1+n/\Rhat)\le2x$. Equations~\eqref{eq:probe-constant}
and~\eqref{eq:cost-constant} therefore give the claimed $2^{-26}$
coefficient in either forest branch.

It remains to consider $x<8192$. Then
$\log(1+n/\Rhat)<8193$. The $n$ initialization inserts, $n$ deletions, and
$n$ replacement inserts all change the represented set and each must perform
a probe. Their $3n$ probes give average $3/4$ over the $4n$ operations, while
$2^{-26}\cdot8193\le3/4$. Hence in every case
\[
  2^{-26}\log\left(1+\frac n\Rhat\right)
  \le \frac{\E[\text{full-run probes}]}{4n}.
\]
This is Theorem~\ref{thm:main}.

\section{Consequences and randomization}\label{sec:consequences}

\begin{proof}[Proof of Theorem~\ref{thm:fusion}]
Canonical strict predecessor answers membership for $x\ne U-1$ because
\[
  x\in S\quad\Longleftrightarrow\quad\operatorname{pred}_{<}(x+1)=x.
\]
The formal reduction preserves the initial memory, represented set, reachable
states, and every finite-batch probe trace. It applies the main theorem on
$[U-1]$. At $U=2^w$, take
$n(w)=\lceil w^c\rceil$ for fixed real $c>0$; for positive integer $c$, take
$n(w)=w^c$ exactly. Then $n(w)^8\le U$ for sufficiently large $w$. A space
bound of
\[
  \left\lceil\log\binom Un\right\rceil+\rho(n),
  \qquad \rho(n)=o(n),
\]
first compiles by Lemma~\ref{lem:physicalization} with
$H=O(\log S)=o(n)$ additional bits. The strict-predecessor reduction then gives
membership redundancy
\[
  R_{<}\le \rho(n)+H+2=o(n).
\]
For an already fixed layout the explicit denominator is $R+2$; for a variable
physical prefix it is $\rho+H+O(1)$. In either case the lower bound diverges.
Theorem~\ref{thm:main} therefore gives $\omega(1)$ expected-amortized time.
A worst-case $O(1)$ structure would have $O(1)$ expected-amortized time, a
contradiction.
\end{proof}

\begin{proof}[Proof of Corollary~\ref{cor:linear}]
If time is at most a constant $T$, Theorem~\ref{thm:main} implies
\[
 R+1\ge n\cdot2^{-C_0T}=\Omega(n).
\]
\end{proof}

\begin{proof}[Proof of Corollary~\ref{cor:scale}]
For sufficiently large $w$, the stated redundancy and Theorem~\ref{thm:main}
give the explicit bound
\[
  \frac{a}{2^{27}}\sqrt{\log w}
  \le \text{expected-amortized cost}.
\]
\end{proof}

\begin{lemma}[Almost-sure schedule totalization]\label{lem:totalization}
Fix finite $U$ and a finite $B$-bit memory layout. A conventional zero-error
Las Vegas cell-probe machine that uses one fresh local random stream per
invocation and terminates almost surely on every valid finite history induces
an operation-indexed family of total deterministic dictionaries and a common
full-measure schedule set on which outputs, persistent transitions, probe
traces, and finite-history costs agree exactly with the source machine.
\end{lemma}
\begin{proof}
For each invocation position, operation, and reachable memory produced by a
finite valid history, termination and exact correctness hold on a
probability-one schedule set. There are countably many such tuples: positions
and histories are countable, while operations and memories are finite. Their
countable intersection is a common full-measure set $G$.

Fix a schedule in $G$, a position, and an operation. Collect the terminating
probe paths only from reachable memories at that position. This is a finite set
of finite paths because the memory domain is finite. Merge their shared
read-value prefixes into one finite interaction trie. At every read value not
used by a reachable path, attach an arbitrary terminating fallback leaf; do the
same after any unrealized write continuation. The resulting total program
agrees exactly with the source machine on every reachable memory. A mixed
unreachable memory used by the communication proof either follows the same
known-cell path as its reachable reference execution, or enters a fallback
branch; the replay lemma uses only the former case. Thus the extension
preserves every source trace needed by the proof without assuming source
termination on corrupt memories. On schedules outside $G$, choose arbitrary
total fallback trees solely to define the family; correctness is required only
almost everywhere and the deterministic theorem is never invoked there.
\end{proof}

For a Las Vegas structure, expose one fresh local stream at every global
primitive-operation position and apply Lemma~\ref{lem:totalization}. Every true
and counterfactual segment at one position uses the same fixed local stream,
while distinct invocations may have different traces. Fix the
input-independent schedule law $\mu$ before the hard instance, giving the
product experiment $\mathcal D_1\times\mu$. The indexed deterministic theorem
holds on $G$; monotonicity of the nonnegative Lebesgue integral preserves its
bound. Independence among schedule positions is unnecessary, although
independent fresh streams are a special case.

A Monte Carlo structure is not a distribution over always-correct
deterministic dictionaries. Our protocol uses exact answers to decode the set
and reject invalid deletion candidates, so per-query error is not covered.

\section{Formal verification and artifact}\label{sec:formal}

The Lean development formalizes the theorem in the operational model of
Definition~\ref{def:model}; it is not an axiomatized restatement of the paper.
In particular, the development defines the packed memory and adaptive
read--write programs, executes the sampled operation sequence, proves the
inner and outer counting bounds, constructs the padded occurrence forest, and
derives the scalar inequalities in Theorems~\ref{thm:main}
and~\ref{thm:fusion}. The canonical adapter module formalizes the charged
physical-prefix state map, mixed-width length field, nested metadata cells,
allocation and release traces, injectivity, and exact trace and probe-cost
simulation. It also formalizes the conventional source-machine totalizer,
finite fallback trie, common almost-everywhere schedule set, mixed-state
replay, exact finite-batch source/target equality, and the indexed-family
lower-bound lift, including almost-everywhere equality between source and
indexed expected-amortized costs. The physical word-RAM
translation premise remains the explicit scope condition in
Lemma~\ref{lem:physicalization}.

\Needspace*{8\baselineskip}
\begin{center}\footnotesize
\begin{tabular}{@{}>{\raggedright\arraybackslash}p{0.33\linewidth}
  >{\raggedright\arraybackslash}p{0.59\linewidth}@{}}
\toprule
Paper statement & Lean declaration \\
\midrule
Theorem~\ref{thm:main}, deterministic &
  \path{indexed_deterministic_expected_amortized_lower_bound} \\
Theorem~\ref{thm:main}, Las Vegas &
  \path{fresh_random_expected_amortized_lower_bound} \\
Theorem~\ref{thm:fusion}, deterministic &
  \path{freshTheorem2_strictDynamicPredecessor_impossibility_realC} \\
Theorem~\ref{thm:fusion}, Las Vegas &
  \path{freshTheorem2_strictDynamicPredecessor_impossibility_lasvegas_realC} \\
Physical-prefix compilation adapter &
  \path{PhysicalPrefixCanonical.run_charged_compile};\newline
  \path{PhysicalPrefixCanonical.probeCount_charged_compile};\newline
  \path{PhysicalPrefixCanonical.PhysicalDictionary.physical_prefix_expected_amortized_lower_bound} \\
Las Vegas source adapter &
  \path{PhysicalPrefixCanonical.Conventional.source_batch_execution_exact};\newline
  \path{PhysicalPrefixCanonical.Conventional.ae_source_batch_probeCount_eq};\newline
  \path{PhysicalPrefixCanonical.Conventional.source_expectedAmortizedCost_eq_indexed};\newline
  \path{PhysicalPrefixCanonical.Conventional.conventional_source_expected_amortized_lower_bound} \\
Corollary~\ref{cor:linear} &
  \path{corollary3_linear_redundancy} \\
Corollary~\ref{cor:scale} &
  \path{corollary4_sqrt_time} \\
\bottomrule
\end{tabular}
\end{center}

The companion artifact places the two proof trees under
\begin{center}\small
\path{LeanProject/SuccinctFusionNodes/}\\
\path{LeanProject/SuccinctFusionNodesFreshRandom/}.
\end{center}
The file
\path{CLAIMS_EVIDENCE.md} maps every load-bearing paper claim to its formal
declarations; \path{Audit.lean} and \path{CanonicalAudit.lean} check the
promoted endpoints; and \path{REPLAY.md} plus \path{CanonicalReplay.md}
record the environment and source identities. From the
packaged \path{formal/lean_project} directory, the complete replay is:

{\small
\begin{verbatim}
nice -n 10 lake -Kjobs=1 build \
  LeanProject
nice -n 10 lake env lean \
  LeanProject/SuccinctFusionNodesFreshRandom/Audit.lean
nice -n 10 lake env lean \
  LeanProject/SuccinctFusionNodes/CanonicalAudit.lean
\end{verbatim}
}

The pinned environment is Lean \texttt{4.32.2}. The Mathlib revision is
\begin{center}\small
\texttt{905b95818eb32af7874a58b427f50c1711a5e96c}.
\end{center}
The aggregate SHA-256 values of the frozen, fresh-random, and canonical
source-adapter Lean trees are, respectively,
\begin{center}\scriptsize
\texttt{e930f8d59a5f6588713ce9b67c5a50f3f8891e21df5d3eed29743f7a39555a55},\\
\texttt{18882fa9b9c71f94ecc3210a2dd85ec53c63439bef8dc392a39e3f519da81d22},\\
\texttt{2aa2cea19a4d08c711c102a1bb1186d97508553353f6b3acc3e079a097d72579}.
\end{center}
The fresh-random top target completed 8,730 jobs. A source scan finds no
\texttt{sorry}, \texttt{admit}, \path{native_decide}, or custom
\texttt{axiom}; the transitive foundations reported for every promoted
endpoint are exactly \texttt{propext}, \texttt{Classical.choice}, and
\texttt{Quot.sound}. The artifact also includes the manuscript source,
rendered PDF, deterministic package builders, and a manifest-checking replay
script. No dataset or specialized hardware is required.

\section{Context and open questions}\label{sec:context}

\Needspace*{9\baselineskip}
\begin{center}\small
\begin{tabular}{@{}p{0.24\linewidth}p{0.29\linewidth}p{0.39\linewidth}@{}}
\toprule
Work & Time in the fusion-node regime & Redundancy or scope \\
\midrule
P\v{a}tra\c{s}cu--Thorup~\cite{PT14} & worst-case $O(1)$ &
  not $o(n)$ bits \\
KLZ~\cite{KLZ26} & optimal amortized expected &
  $o(n)$ only in their polynomial-universe regime \\
Domingues~\cite{D26} & worst-case $O(1/\epsilon)$ &
  $O(nw^\epsilon)$ bits; smaller result is a rank index \\
LLYZ~\cite{LLYZ23} & lower bound &
  stated for $U=n^{1+\Theta(1)}$ \\
This paper & lower bound &
  $2^{-26}\log(1+n/(R+1))$ for $n^8\le U$ and $\log U\ge2^{70}$ \\
\bottomrule
\end{tabular}
\end{center}

The remaining questions include:
\begin{enumerate}
  \item Can a full dictionary achieve redundancy
        $n\cdot2^{-\Theta(\sqrt{\log w})}$ and time $O(\sqrt{\log w})$?
  \item Does a comparable lower bound hold for Monte Carlo dictionaries?
  \item Can the update-only LLYZ lower bound be extended to large universes?
  \item What is the exact constant-time frontier between $\Theta(n)$ and
        Domingues' $O(nw^\epsilon)$ redundancy?
\end{enumerate}

\appendix
\section{Finite coding and forest details}\label{app:details}

This appendix expands the proof joints that are compact in the main text. It
uses no additional assumptions.

\subsection{Adaptive Kraft coding in the outer game}

Fix the public outer context and a payload $y=(C_{\mathrm{end}},X)$, where
$X\in\{0,1\}^{\lambda}$ is the cap vector. Let $\mathcal P_y$ be the ordered
partitions compatible with $y$. After blocks $A^{[1]},\ldots,A^{[i-1]}$ have
been decoded, let $F_i(y,A^{[<i]})$ be the family used for the next rank. It is
either the complete remaining $m$-set family or the qualified family, and it
contains the true next block.

For $p\in\mathcal P_y$, define
\[
 q_y(p):=\prod_{i=1}^{\lambda}
   |F_i(y,A^{[<i]})|^{-1},
 \qquad
 c_y(p):=-\log q_y(p)=\sum_i\log|F_i(y,A^{[<i]})|.
\]
The recursive family tree gives
\begin{equation}\label{eq:kraft-fiber}
  \sum_{p\in\mathcal P_y}q_y(p)\le1.
\end{equation}
Indeed, at each node the first-block weights sum to one, and induction applies
to every child source. The log-sum inequality, or equivalently the entropy
converse to~\eqref{eq:kraft-fiber}, therefore gives
\[
  \frac1{|\mathcal P_y|}\sum_{p\in\mathcal P_y}c_y(p)
  \ge \log|\mathcal P_y|.
\]
Thus no integer ceiling is charged per stage. Averaging over payload fibers,
the cap vector contributes exactly $\lambda$ bits. The final-memory fiber
contributes at most $R$ bits on average by the argument below. This yields
Equation~\eqref{eq:outer-lower}; Equation~\eqref{eq:outer-upper} follows by
inserting the localized-family savings into the exact ideal cost $c_y(p)$.

\subsection{Endpoint fibers}

Let $\mathcal C_T$ be the full-capacity memory strings representing target set
$T$. Exact reachability makes every $n$-set attainable, while
\[
  \sum_T|\mathcal C_T|\le2^B.
\]
Since the endpoint target is uniform, concavity of the logarithm gives
\[
 \frac1{\binom Un}\sum_T\log|\mathcal C_T|
 \le \log\frac{2^B}{\binom Un}=R.
\]
A self-delimiting integer rank adds less than one bit, giving the paper's
$R+1$ reserve. This argument counts unreachable memory strings only in the
upper bound $2^B$; it never treats them as legal states.

\subsection{Residual pair law}

Fix $K$, intersection size $g$, and the residual insertion superset available
before segment $i$. For every pair of $m$-sets $A,Q\subseteq[U]\setminus K$
with $|A\cap Q|=g$, the union has size $2m-g$. The number of residual supersets
of the prescribed size containing that pair is
\[
  \binom{|[U]\setminus K|-(2m-g)}
        {|\mathcal A_{\mathrm{rest}}|-(2m-g)},
\]
which is independent of the pair. Hence residual-first sampling followed by a
uniform candidate produces exactly the inner game's conditional pair law.
The deletion and insertion relative permutations are sampled independently of
these semantic sets and are then fixed before the inner probability bound.

\subsection{Unique-LCA forest accounting}

Write the actual full probe trace as an ordered list of address occurrences.
For every nonfirst occurrence of an address, pair it with the immediately
preceding occurrence of that address. At one padded hierarchy node, assign the
later occurrence to the unique lowest common ancestor whose distinct children
contain the pair. If an address occurs in $s$ child supports, its consecutive
cross-child pairs contribute exactly $s-1$ units, which equals the node charge.

Each later occurrence has one immediate predecessor and one lowest common
ancestor, so it is charged at most once over the entire hierarchy. Padded empty
leaves cover out-of-range positions. A partial node receives no local charge,
but recursion continues into it, preserving every smaller full node in the
leftover interval. Therefore the sum of all selected full-node charges is at
most the number of noninitial occurrences in the one actual trace, and hence
at most its length. At a fixed level, the full nodes are disjoint aligned
windows; their local charges are exactly the repeated-window costs used by
Lemma~\ref{lem:outer}. These two facts give the all-level inequality used in
Section~\ref{sec:main-proof}.

\section*{Research-system disclosure}
This work was performed with material assistance from Beyond, the research
system operated by Nth Research Collective. The author reviewed the complete
proof, accepts responsibility for the claims, and made the final scientific
and editorial decisions.

\end{document}